\documentclass[letterpaper]{article} 
\usepackage{aaai21}  
\usepackage{times}  
\usepackage{helvet} 
\usepackage{courier}  
\usepackage[hyphens]{url}  
\usepackage{graphicx} 
\usepackage{graphicx}  
\usepackage{natbib}  
\usepackage{caption} 
\usepackage{amsfonts}
\usepackage{amssymb}
\usepackage{lineno}

\usepackage{graphicx}
\usepackage{multirow}
\usepackage{multicol}
\usepackage{algorithm}
\usepackage{algorithmic}
\usepackage{amsthm}
\usepackage{amsmath}
\usepackage{amsfonts}

\newcommand{\eat}[1]{}

\newtheorem{definition}{Definition}
\newtheorem{problem}{Problem}

\newtheorem{assumption}{Assumption}

\DeclareMathOperator*{\argmin}{arg\,min} 

\newtheorem{theorem}{Theorem}
\makeatletter
\newcommand*{\indep}{%
  \mathbin{%
    \mathpalette{\@indep}{}%
  }%
}
\newcommand*{\nindep}{%
  \mathbin{
    \mathpalette{\@indep}{\not}
  }%
}
\newcommand*{\@indep}[2]{%
  \sbox0{$#1\perp\m@th$}
  \sbox2{$#1=$}
  \sbox4{$#1\vcenter{}$}
  \rlap{\copy0}
  \dimen@=\dimexpr\ht2-\ht4-.2pt\relax
  \kern\dimen@
  {#2}%
  \kern\dimen@
  \copy0 
}

\title{Sufficient Dimension Reduction for Average Causal Effect Estimation}
\author {Debo Cheng\thanks{chedy055@mymail.unisa.edu.com}, Jiuyong Li\thanks{Jiuyong.Li@unisa.edu.au}, Lin Liu, Jixue Liu\\}

\affiliations {University of South Australia, Mawson Lakes, SA 5095, Australia. 
}

\begin{document}
\maketitle

\begin{abstract}
Having a large number of covariates can have a negative impact on the quality of causal effect estimation since confounding adjustment becomes unreliable when the number of covariates is large relative to the samples available. Propensity score is a common way to deal with a large covariate set, but the accuracy of propensity score estimation (normally done by logistic regression) is also challenged by large number of covariates. In this paper, we prove that a large covariate set can be reduced to a lower dimensional representation which captures the complete information for adjustment in causal effect estimation. The theoretical result enables effective data-driven algorithms for causal effect estimation. We develop an algorithm which employs a supervised kernel dimension reduction method to search for a lower dimensional representation for the original covariates, and then utilizes nearest neighbor matching in the reduced covariate space to impute the counterfactual outcomes to avoid large-sized covariate set problem. The proposed algorithm is evaluated on two semi-synthetic and three real-world datasets and the results have demonstrated the effectiveness of the algorithm.
\end{abstract}

\section{Introduction}
\label{Sec:Intro}
Estimating the causal effect of an action (also known as exposure, intervention or treatment in literature) on an outcome is a central problem in scientific discovery and it is the essential requirement for decision making in many areas such as medical treatments, government policy marking, and marketing, to name a few.

A key to accurate causal effect estimation is confounding control as uncontrolled confounding variables can introduce spurious association between the treatment and the outcome, biasing the estimation of causal effect. Properly designed and executed randomization in a randomized control trial (RCT) is the gold standard for confounding control~\cite{deaton2018understanding}. However, it is often impossible or too costly to conduct RCTs. As a result, it is desirable to estimate causal effects from observational data, and data-driven causal effect estimation has attracted much attention in recent years~\cite{imbens2015causal,haggstrom2018data}.


Controlling or adjusting for a \emph{deconfounding set} (also known as adjustment set in literature) is an effective way to eliminate confounding bias in causal effect estimation using observational data~\cite{pearl2009causality,sekhon2011multivariate,vanderweele2011new,shpitser2012validity}.
The size of the deconfounding set can significantly affect the performance of a causal effect estimator~\cite{abadie2006large,benkeser2017doubly}, and a small sized deconfounding set is preferred~\cite{de2011covariate,witte2019covariate}.

There exist two general approaches for determining a proper deconfounding set, each with own limitations: \textbf{1.} Including all covariates in the deconfounding set. This is a straightforward approach, but a large sized deconfounding set leads to the reduction of statistical gain problem~\cite{de2011covariate}; \textbf{2.} Selecting a subset of covariate variables to form the deconfounding set, based on some criterion, mostly, the back-door criterion or its variations~\cite{pearl2009causality,maathuis2015generalized}. However, the underlying causal graphs required by these criteria are usually unknown and it is impossible to recover a unique causal graph from the data alone.

Another line of research is focused on dimension reduction techniques to create a small set of variables in a different space for confounding adjustment. An early and notable example of this type of techniques is propensity score~\cite{rubin1974estimating,rosenbaum1983central}, which reduces a covariate set to a scalar, specifically, the probability of an individual receiving the treatment given the covariates. However, propensity score estimation again suffers from large-sized covariate sets~\cite{hahn1998role,van2014entering,luo2017estimating}. More recently, some advancement has been made along the direction of dimension reduction for causal effect estimation (details in the Related work section). However, it is not clear whether or not dimension reduction guarantees unbiased causal effect estimation.

In this paper, we prove that the deconfounding set obtained under Sufficient Dimension Reduction (SDR) is sufficient to control confounding bias, based on the graphical causal model. This result opens a door for developing new methods for causal effect estimation with a large number of covariates. We then propose a method CESD, the \underline{C}ausal \underline{E}ffect estimator by using \underline{S}ufficient \underline{D}imension reduction. This method utilizes kernel dimension reduction~\cite{fukumizu2004dimensionality} which satisfies the SDR conditions to learn a deconfounding set from data and captures the conditional independence on covariance operators using the \emph{reproducing kernel Hilbert spaces}~\cite{aronszajn1950theory,hofmann2008kernel}. The main technical contributions of the work can be summarized as follows.

\begin{itemize}
  \item We have developed a theorem to show that the deconfounding set obtained under SDR is sufficient for controlling confounding bias in causal effect estimation based on graphical causal modeling. To the best of our knowledge, this is the first work which proves that the reduced covariate set by SDR is a proper deconfounding set.
  \item With the support of the theorem, we develop a data-driven algorithm, CESD, which finds a deconfounding set satisfying the conditional independence in RKHS without previous assumptions and utilizes nearest neighbour matching with the deconfounding set for average causal effect estimation.
  \item The experimental results on two semi-synthetic and three real-world datasets have demonstrated the effectiveness of CESD in causal effect estimation, compared with state-of-the-art methods. The experiments also investigate and demonstrate the superiority of the deconfounding set found by CESD over propensity score.
\end{itemize}

\section{Related work}
\label{Sec:Relatedworks}
Our work is closely related to representation learning for causal effect estimation, which aims to transform from the original covariate space to a new representation space. The learned representation set or reduced covariate set is used in various ways in causal effect estimation, including for propensity score estimation, outcome regression or distribution balancing. In the following, we review the related work based on the way of using the learned representation set. 

A doubly robust estimator makes use of propensity score and outcome regression to reduce possible misspecification of one model for causal effect estimation.~\cite{van2006targeted,funk2011doubly}. Sufficient dimension reduction methods have recently attracted attention in improving the performance of doubly robust estimators~\cite{liu2018alternative,ma2019robust,ghosh2020sufficient}. Liu et al. adopted sufficient dimension reduction for predicting
propensity score only~\cite{liu2018alternative}. Ma et al. utilized sparse sufficient dimension reduction to estimate the propensity score and recover the outcome model~\cite{ma2019robust}. Ghosh et al. considered efficient semi-parametric sufficient dimension reduction methods in all nuisance models, and then combined these into classical imputation and inverse probability weighting (IPW) estimators~\cite{ghosh2020sufficient}. However, doubly robust estimators require the specific parameters of models, and these parameters are likely to be inconsistently estimated~\cite{benkeser2017doubly}.

The most relevant work to ours is the matching method developed by Luo et al.~\cite{luo2019matching}. The work considered sufficient dimension reduction for building models on sub-datasets containing the treated samples and the control samples to construct two low-dimensional representation sets as the balance scores for matching, but not for identifying a deconfounding set. When the number of samples in a dataset is small, dividing it into two sub-datasets will reduce the performance of the sufficient dimension reduction method. 

Recently, a number of deep learning methods have been developed for causal effect estimation from observational data~\cite{shalit2017estimating,yao2018representation,shi2019adapting}. With these methods, the learning of representation set aims to balance the distributions of the treated and control groups. The advantage of deep learning methods is that they can capture complex nonlinear representations and handle high-dimensional data with large sample size. However, massive parameter turning is very difficult, and low interpretability limits their applications.

Additionally, many machine learning models have been designed for causal inference such as trees-based methods~\cite{hill2011bayesian,athey2016recursive,kunzel2019metalearners}, and re-weighting methods~\cite{rosenbaum1983central,zubizarreta2015stable,kuang2017estimating}. Since they do not involve representation learning, these methods are not directly related to our work in this paper.

Compared with the above works, our work provides a theorem to ensure that the set of the reduced covariates is a deconfounding set for unbiased causal effect estimation. As far as we know, this conclusion has not been reported before. This theorem enables to directly control confounding bias in a low dimensional space.

\section{Notations and assumptions}
\label{subsec:notations}
We consider a binary treatment variable $W$ ($W=1$ for treated and $0$ for control). The potential outcomes $Y(w)$ is relative to a specific treatment $W=w (w\in \{0, 1\})$. For each sample (individual) $i$, there is a pair of potential outcomes, $(Y_{i}(0), Y_{i}(1))$. Only one of the potential outcomes can be observed, and the other one is counterfactual~\cite{rubin1974estimating,robins1986new}. 
We use $Y_i \in \mathbb{R}$ to denote observed outcome of sample $i$, and we have $Y_{i}=w_i*Y_{i}(1)+(1-w_i)*Y_{i}(0)$. We omit the subscript $i$ when the meaning is clear. 

Let $\mathbf{X}\in\mathbb{R}^{p\times 1}$ be a set of pretreatment variables, where $p$ denotes the dimensions of variables. We make the pretreatment assumption, i.e. each variable in $\mathbf{X}$ is measured before assigning the treatment variable $W$ and observing the response variable $Y$. This is a realistic assumption as it reflects how a sample is obtained in many application areas such as economics and epidemiology~\cite{hill2011bayesian,imbens2015causal,abadie2016matching}.
Given a dataset $\mathbf{D}$ containing $n$ samples of $(\mathbf{X}, Y)$, the average causal effect ($ACE$) and average causal effect on the treated samples ($ACT$) can be estimated by the following equations respectively.

\begin{equation}
\label{eq:ACE}
\begin{split}
ACE  &= \mathbb{E}[Y(1)-Y(0)] = \sum_z[\mathbb{E}(Y\mid w, \mathbf{Z}=z)\\
&-\mathbb{E}(Y\mid w', \mathbf{Z}=z)]Pr(\mathbf{Z}=z)
\end{split}
\end{equation}

\begin{equation}
\label{eq:ACT}
\begin{split}
ACT  &= \mathbb{E}[Y(1)-Y(0)\mid w]  \\
 &= \sum_z[\mathbb{E}(Y\mid w, \mathbf{Z}=z)]Pr(\mathbf{Z}=z)
\end{split}
\end{equation}
\noindent where $w$, $w'$ and $\mathbb{E}(\cdot)$ refer to $W=1$, $W=0$ and the expectation function, respectively. $\mathbf{Z}$ is a deconfounding set and is what we focus on in this paper. To estimate $ACE$ or $ACT$ from observational data, we make the following two assumptions which are commonly used in causal inference literature~\cite{imbens2015causal}.

\begin{assumption}[unconfoundedness]
\label{assum:001}
The potential outcomes are independent of the treatment variable $W$ given all the other variables $\mathbf{X}$. Formally, $(Y(0), Y(1))\indep W|\mathbf{X}$.
\end{assumption}

\begin{assumption}[Overlap]
\label{assum:002}
Every sample has a nonzero probability to receive treatment $1$ or $0$ when conditioned on the pretreatment variables $\mathbf{X}$, i.e. $0 < Pr(W=1|\mathbf{X}) < 1$.
\end{assumption}

The unconfoundedness assumption means that there is ``\emph{no hidden confounder}'' in the system. The purpose of this paper is to find a deconfounding set $\mathbf{Z}$ such that $(Y(0), Y(1))\indep W|\mathbf{Z}$ holds, i.e. the spurious association between $W$ and $Y$ are blocked by the set $\mathbf{Z}$. In this paper, we use a \emph{causal graphical model} when discovering a deconfounding set $\mathbf{Z}$ from observational data.

A directed acyclic graph (DAG) $\mathcal{G}$ is a graph that includes directed edges and does not contain directed cycles. In a DAG $\mathcal{G}$, a path is a sequence of consecutive edges. A directed edge ``$\rightarrow$'' denotes a cause-effect relationship, e.g. $X_i\rightarrow X_j$ indicates that $X_i$ is a direct cause (or parent) of $X_j$, and equivalently $X_j$ is a direct effect (or child) of $X_i$. A node $X_i$ is a collider if there are two (or more) edges pointing to it, i.e. $\rightarrow X_i \leftarrow$. The independencies between variables in a DAG can be read off the DAG $\mathcal{G}$ based on $d$-separation, as defined as follows.

\begin{definition}[$d$-separation~\cite{pearl2009causality}]
    \label{d-separation}
    A path $\pi$ in a DAG $\mathcal{G}$ is said to be $d$-separated (or blocked) by a set of nodes $\mathbf{Z}$ if and only if (1) $\pi$ contains a chain $X_i \rightarrow X_k \rightarrow X_j$ or a fork $X_i \leftarrow X_k \rightarrow X_j$ node such that the middle node $X_k$ is in $\mathbf{Z}$, or (2) $\pi$ contains a collider $X_k$ such that $X_k$ is not in $\mathbf{Z}$ and no descendant of $X_k$ is in $\mathbf{Z}$.
\end{definition}

When a DAG $\mathcal{G}$ is given, the back-door criterion can be used to determine if $\mathbf{Z}\subseteq\mathbf{X}$ is sufficient for identifying the causal effects of $W$ on $Y$~\cite{pearl2009causality}. 

\begin{definition}[Back-door criterion]
  \label{def:backdoorcrite}
For an ordered pair of variables $(W, Y)$, a set of variables $\mathbf{Z}$ is said to satisfy the back-door criterion in a given DAG $\mathcal{G}$ if
\begin{enumerate}
  \item $\mathbf{Z}$ does not contain a descendant node of $W$;
  \item $\mathbf{Z}$ blocks every back-door path between $W$ and $Y$ (i.e. paths between $W$ and $Y$ containing an arrow into $W$).
\end{enumerate}

\end{definition}
If we can find a set of variables $\mathbf{Z}$ which satisfies the back-door criteria, then $\mathbf{Z}$ is a proper deconfounding set or adjustment set, and $ACE$ (or $ACT$) can be estimated from data by adjusting for $\mathbf{Z}$ as shown in Eq.(\ref{eq:ACE}) (or Eq.(\ref{eq:ACT})). In order to describe how to identify a deconfounding set $\mathbf{Z}$, we need to use a manipulated DAG.


\begin{definition}[Manipulated DAG $\mathcal{G}_{\underline{W}}$]
The graph $\mathcal{G}_{\underline{W}}$ is a manipulated DAG of the DAG $\mathcal{G}$ when all edges outgoing from $W$ are removed from $\mathcal{G}$.
\end{definition}

Based on the above definition, in the manipulated DAG $\mathcal{G}_{\underline{W}}$ all directed paths from $W$ to $Y$ have been removed and only all back-door paths between $W$ and $Y$ are retained. Hence, a set $\mathbf{Z}$ that $d$-separates $W$ and $Y$ in $\mathcal{G}_{\underline{W}}$ will block all back-door paths between $W$ and $Y$.

\subsection{Problem setup}
\label{subsec:prosol}
In this paper, we aim at searching for a deconfounding set $\mathbf{Z}$ which is a low-dimensional representation of the set of covariates $\mathbf{X}$. The problem definition is given as follows.

\begin{problem}
\label{def:prob}
We convert the problem of determining a deconfounding set from the original covariate space $\mathbf{X}$ to the problem of learning a low-dimensional representation set $\mathbf{Z}$ from $\mathbf{X}$ such that $W\indep Y\mid \mathbf{Z}$ in the manipulated DAG $\mathcal{G}_{\underline{W}}$.
\end{problem}

When the deconfounding set $\mathbf{Z}$ is found, the causal effect of $W$ on $Y$ can be estimated unbiasedly by adjusting for $\mathbf{Z}$ as in Eq.(\ref{eq:ACE}) (or Eq.(\ref{eq:ACT})).

\section{Theory and algorithm}
\label{sec:theo}
In this section, we first prove that the reduced covariates set $\mathbf{Z}=\mathbf{\Psi}^{T}\mathbf{X}$ by \emph{sufficient dimension reduction} (SDR) is sufficient to remove confounding bias in causal effect estimation. Then we presented the CESD algorithm.

\subsection{Sufficient condition for identifying a deconfounding set}
\label{subsec:TheoreticalS}
Let's consider the treatment assignment $W$ as a binary classification problem, i.e. the probability density function of $W$ given $\mathbf{X}$ is $Pr_{W|\mathbf{X}}(w|x)$.
SDR attempts to search for a projection $\mathbf{\Psi}\in\mathbb{R}^{p\times r}$, where $r<p$, such that
\begin{equation}
\label{eq:conind}
  W\indep \mathbf{X} | \mathbf{\Psi}^{T}\mathbf{X}
\end{equation}
\noindent where $\mathbf{\Psi}^{T}\mathbf{X}$ is the orthogonal projection of $\mathbf{X}$ onto the column subspace of $\mathbf{\Psi}$, and the column subspace of $\mathbf{\Psi}$ is refer to the \emph{dimension reduction space} (DRS)~\cite{cook1996graphics,cook2009regression}. The primary interest is the \emph{central} DRS since it has a well-known invariance property~\cite{cook1996graphics,cook2009regression}.

\begin{definition}[Central DRS~\cite{cook1996graphics}]
The column space of $\mathbf{\Psi}$ is a central DRS if the column space of $\mathbf{\Psi}$ is a DRS with the smallest possible dimension $r$.
\end{definition}

\begin{figure}[t]
  \centering
  \includegraphics[width=5.2cm,height=2.7cm]{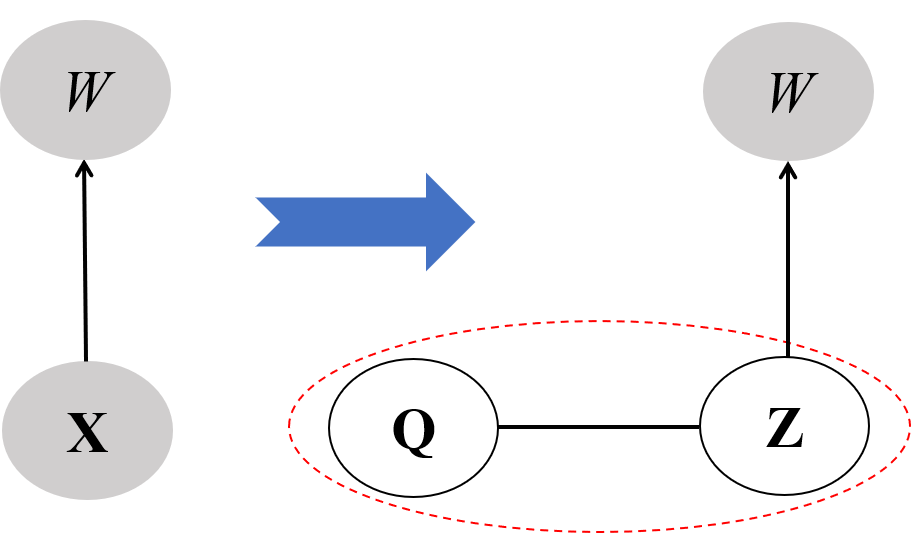}
  \caption{A graphical representation of sufficient dimension reduction, $W\indep \mathbf{Q} | \mathbf{Z}$ holds~\cite{fukumizu2004dimensionality}, where $\mathbf{X}$  is decomposed into $(\mathbf{Z, Q})$.}
  \label{fig:SDR}
\end{figure}

Identifying a projection $\mathbf{\Psi}$ makes Eq.(\ref{eq:conind}) hold is equivalent to searching for a projection $\mathbf{\Psi}$ which makes $W$ and $(\mathbf{I}-\mathbf{\Psi} ^{T})\mathbf{X}$ conditionally independent given $\mathbf{\Psi}^{T}\mathbf{X}$. That is, $\mathbf{X}$ can be decomposed into $(\mathbf{Z=\mathbf{\Psi}^{T}\mathbf{X}, \mathbf{Q}=(\mathbf{I}-\mathbf{\Psi} ^{T})\mathbf{X}})$, $\mathbf{Z}$ which is associated with $W$ and $\mathbf{Q}$ which is independent of $W$ given $\mathbf{Z}$, i.e. $W\indep \mathbf{Q} | \mathbf{Z}$ (See Fig.~\ref{fig:SDR}).

Now we show in the following theorem that finding a deconfounding set from $\mathbf{X}$ can be converted to the problem of learning the reduced covariate set $\mathbf{Z}$ by SDR.
\begin{theorem}
  \label{theo:sdrBD}
    Let $\mathbf{D}$ be a dataset that contains the treatment variable $W$, the outcome variable $Y$, and a set of all pretreatment variables $\mathbf{X}$, and $\mathcal{G}$ be the DAG representing the underlying causal structure (relationships) of the covariates. If there exists a central DRS (the column space of $\mathbf{\Psi}$) such that $W\indep \mathbf{Q}\mid \mathbf{Z}$, where $\mathbf{Z}= \mathbf{\Psi}^{T}\mathbf{X}$ and $\mathbf{Q}=(\mathbf{I}-\mathbf{\Psi} ^{T})\mathbf{X}$, then $\mathbf{Z}$ is a proper deconfounding set for estimating the causal effect of $W$ on $Y$ unbiasedly.
\end{theorem}
\begin{proof}
    Under the pretreatment assumption of $\mathbf{X}$, there is no descendant node of $W$ in $\mathbf{X}$. Hence, there is a directed path $W\rightarrow Y$ in the DAG $\mathcal{G}$. Under the unconfoundedness assumption, there are no hidden confounders, and in the manipulated DAG $\mathcal{G}_{\underline{W}}$, there are only back-door paths between $W$ and $Y$. Hence, all back-door paths between $W$ and $Y$ are blocked by the set of $\mathbf{X}$, i.e. $W\indep Y|\mathbf{X}$ holds in $\mathcal{G}_{\underline{W}}$.

    We now prove that if $\mathbf{Z}= \mathbf{\Psi}^{T}\mathbf{X}$ satisfies $W\indep \mathbf{Q} | \mathbf{Z}$, then $\mathbf{Z}$ is sufficient to block all block-door paths between $W$ and $Y$. We have $W\indep Y|\mathbf{X}$ in $\mathcal{G}_{\underline{W}}$ and $\mathbf{X}=(\mathbf{Z, Q})$, then $W\indep Y|(\mathbf{Z, Q})$ holds in $\mathcal{G}_{\underline{W}}$. As $W\indep \mathbf{Q} | \mathbf{Z}$ is satisfied, $W\indep (Y, \mathbf{Q})|\mathbf{Z}$ holds in $\mathcal{G}_{\underline{W}}$ by the contraction property of conditional independence. According to the decomposition property of the conditional independence, $W\indep (Y, \mathbf{Q})|\mathbf{Z}$ is sufficient to infer $W\indep Y|\mathbf{Z}$ in $\mathcal{G}_{\underline{W}}$. As there are only back-door paths between $W$ and $Y$ in $\mathcal{G}_{\underline{W}}$, $\mathbf{Z}$ is sufficient to block all such paths since $W\indep Y|\mathbf{Z}$ holds in $\mathcal{G}_{\underline{W}}$. Therefore, $\mathbf{Z}$ is a proper deconfounding set. 
\end{proof}

Theorem~\ref{theo:sdrBD} has shown that the reduced covariates set $\mathbf{Z}$ by SDR is sufficient to remove confounding bias when estimating the causal effects of $W$ on $Y$ from data.

\subsection{Deconfounding set identification using SDR}
\label{subsec:kernelmethod}
In this section, we use a kernel-based SDR method to identify a deconfounding set from data. We utilize the cross-covariance operators on \emph{reproducing kernel Hilbert space} (RKHS) ~\cite{aronszajn1950theory} $\mathcal{H}$ to establish an objective function for dimensionality reduction. By default, we use the \emph{Gaussian kernel}, i.e. $k(x_i, x_j) = exp(-\frac{\parallel x_i - x_j\parallel^{2}}{2\delta^{2}})$, where $\delta$ is the kernel width.

Given two RKHS, $(\mathcal{H}_1, k_1)$ and $(\mathcal{H}_2, k_2)$ which are over the measured spaces $(\Omega_1, \mathfrak{B}_1)$ and $(\Omega_2, \mathfrak{B}_2)$, with the positive kernels $k_1$, $k_2$ measurable. For the pair of $W$ and the set of $\mathbf{X}$ on $\Omega_1 \times\Omega_2$, the cross-covariance operator from $\mathcal{H}_1$ to $\mathcal{H}_2$ is defined by the relation:
\begin{equation}
\label{eq:cco}
\begin{aligned}
<g,\Sigma_{W\mathbf{X}}f>_{\mathcal{H}_2} =
 \mathbb{E}_{\mathbf{X}W}[f(\mathbf{X})g(W)] -& \\ \mathbb{E}_{\mathbf{X}}[f(\mathbf{X})]\mathbb{E}_{W}[g(W)]
\end{aligned}
\end{equation}
for all $f\in \mathcal{H}_1$ and $g\in \mathcal{H}_2$.
Hence, the conditional covariance operator $\Sigma_{WW|\mathbf{Z}}$ on $\mathcal{H}_1$ can be defined as follows.

\begin{equation}
\label{eq:007}
\Sigma_{WW|\mathbf{Z}} := \Sigma_{WW} - \Sigma_{W\mathbf{Z}}\Sigma_{\mathbf{Z}\mathbf{Z}}^{-1}\Sigma_{\mathbf{Z}W}
\end{equation}

Theorem 7 in~\cite{fukumizu2004dimensionality} has shown that $\Sigma_{WW|\mathbf{Z}}\geq \Sigma_{WW|\mathbf{X}}$ for any $\mathbf{Z}$, and $\Sigma_{WW|\mathbf{X}} - \Sigma_{WW|\mathbf{Z}} =0 \Leftrightarrow W\indep \mathbf{Q} | \mathbf{Z}$. That is, minimizing $\hat{\Sigma}_{WW|\mathbf{Z}}$ is the most informative direction for searching the optimal $\mathbf{Z}$. Hence, searching for a set of the reduced covariates $\mathbf{Z}$ such that $W\indep \mathbf{Q}|\mathbf{Z}$ holds is equivalent to optimize the minimized conditional covariance operator $\hat{\Sigma}_{WW|\mathbf{Z}}$. The determinant of $\hat{\Sigma}_{WW|\mathbf{Z}}$ can be written as follows. 

\begin{equation}
\label{eq:009}
  \det \hat{\Sigma}_{WW|\mathbf{Z}} = \frac{\det \hat{\Sigma}_{[W\mathbf{Z}][W\mathbf{Z}]}}{\det \hat{\Sigma}_{\mathbf{Z}\mathbf{Z}}}
\end{equation}

\noindent where $\hat{\Sigma}_{[W\mathbf{Z}][W\mathbf{Z}]} = \begin{pmatrix}
\hat{\Sigma}_{WW} & \hat{\Sigma}_{W\mathbf{Z}}\\
\hat{\Sigma}_{\mathbf{Z}W} & \hat{\Sigma}_{\mathbf{Z}\mathbf{Z}}
\end{pmatrix} = \\
\begin{pmatrix}
(\hat{K}_{W} + \epsilon \mathbf{I}_{n})^{2} & \hat{K}_{W}\hat{K}_{\mathbf{Z}} \\
  \hat{K}_{\mathbf{Z}}\hat{K}_{W}  & (\hat{K}_{\mathbf{Z}} + \epsilon \mathbf{I}_{n})^{2}
\end{pmatrix}$, where $\epsilon$ is a positive regularization parameter. $\hat{K}_{W}$ and $\hat{K}_{\mathbf{Z}}$ are the centralized \emph{Gram matrices} defined as follows.
\begin{equation}\label{eq:010}
\hat{K}_{W} = (\mathbf{I}_{n} - \frac{1}{n}\mathbf{1}_{n}\mathbf{1}^{T}_{n})G_{W}(\mathbf{I}_{n} - \frac{1}{n}\mathbf{1}_{n}\mathbf{1}^{T}_{n})
\end{equation}
\begin{equation}\label{eq:011}
\hat{K}_{\mathbf{Z}} = (\mathbf{I}_{n} - \frac{1}{n}\mathbf{1}_{n}\mathbf{1}^{T}_{n})G_{\mathbf{Z}}(\mathbf{I}_{n} - \frac{1}{n}\mathbf{1}_{n}\mathbf{1}^{T}_{n})
\end{equation}
\noindent where $(G_{W})_{i,j} = k(w_i, w_j)$, $(G_{\mathbf{Z}})_{i,j} = k(\mathbf{Z}_i, \mathbf{Z}_j)$ and $\mathbf{1}_n=(1,\dots ,1)^{T}$ is a vector with all elements equal to 1.

To solve Eq.(\ref{eq:009}), gradient descent with line search can be used. The matrix of parameters is updated iteratively by

\begin{equation}
\begin{aligned}
\mathbf{\Psi}^{t+1} &= \mathbf{\Psi}^{t} - \beta \frac{\partial \log \det \hat{\Sigma}_{WW|\mathbf{Z}}}{\partial \mathbf{\Psi}} \\
&= \mathbf{\Psi}^{t} -\beta Tr[\hat{\Sigma}^{-1}_{WW|\mathbf{Z}}\frac{\partial\hat{\Sigma}_{WW|\mathbf{Z}}}{\partial \mathbf{\Psi}}]
\end{aligned}
\label{eq:012}
\end{equation}

\noindent where the trace norm in Eq.(\ref{eq:012}) can be rewritten as $2\epsilon Tr[\hat{\Sigma}^{-1}_{WW|\mathbf{Z}}\hat{K}_{W}(\hat{K}_{\mathbf{Z}}+\epsilon\mathbf{I}_{n})^{-1}\frac{\partial\hat{K}_{\mathbf{Z}}}{\partial\mathbf{\Psi}}(\hat{K}_{\mathbf{Z}}+\epsilon\mathbf{I}_{n})^{-2}\hat{K}_{\mathbf{Z}}\hat{K}_{W}]$.
All of these matrices in Eq.(\ref{eq:012}) can be obtained directly based on Eq.(\ref{eq:010}) and Eq.(\ref{eq:011}). Therefore, the problem of identifying a deconfounding set $\mathbf{Z}$ can be achieved by optimizing Eq.(\ref{eq:009}).

\subsection{NNM using the discovered deconfounding set}
\label{subsec:algCESD}
Given the learned deconfounding set $\mathbf{Z}$, our next step is to infer the counterfactual outcome denoted as $Y^{*}_i(w_i)$. Nearest Neighbor Matching (NNM) is a well-known method for such inference~\cite{rubin1973matching,abadie2006large}. With NNM, the unobserved outcome or counterfactual outcome of an individual $i$ is imputed by the observed outcome of an individual $j$ who has the similar covariates ($\mathbf{Z}$ values) in the opposite treatment group. The Mahalanobis distance is used to measure the distance of each pair $(z_i, z_j)$ over the space of the deconfounding set $\mathbf{Z}$ as follows.

\begin{equation}
\label{eq:mah}
Dist(z_i, z_j) = \{(z_i - z_j)^{T}\hat{\Sigma}_{z}^{-1}(z_i - z_j)\}^{\frac{1}{2}}
\end{equation}

$z_i$ and $z_j$ are the value vector of the deconfounding set $\mathbf{Z}$ of the $i$-th and $j$-th individuals. The strategy of NNM can be formalized as follows.

\begin{equation}
\label{eq:NNM}
Y^{*}_{i}(w_i) = Y_{k}(1-w_{i}); \quad k = \argmin_{j\in\mathbf{D}_{(1-w_{i})}} Dist(z_i, z_j)
\end{equation}

\noindent where $\mathbf{D}_{(1-w_i)}$ is the dataset with the treatment of $1-w_i$.

The complete procedure of CESD is listed in Algorithm~\ref{pseudocode01}.


\textbf{Complexity analysis}: Three parts contribute to the time complexity of CESD. The calculation of $\det \hat{\Sigma}_{WW|\mathbf{Z}}$ is matrix multiplication which has time complexity of $\textbf{O}(np^{2})$. Solving Eq.(\ref{eq:009}) requires a linear search, i.e. $\textbf{O}(np)$. The calculation of NNM is $\textbf{O}(nr^{2})$. Therefore, the time complexity of CESD is $\textbf{O}(np^{2})$.

\begin{algorithm}[tp]
\caption{\underline{C}ausal \underline{E}ffect estimator by using \underline{SD}R (CESD)}
\label{pseudocode01}
\small
\begin{flushleft}
\noindent {\textbf{Input}}: Dataset $\mathbf{D}$ with $W$, $Y$ and pretreatment variables $\mathbf{X}$. The parameters $\epsilon$, $\delta$, the maximum number of Iteration $Ite$ and the dimension of the reduced covariates $r$. \\
\noindent {\textbf{Output}}: Causal effect
\end{flushleft}
\begin{algorithmic}[1]
\STATE {Computing Eq.(\ref{eq:010}) and Eq.(\ref{eq:011})}
\STATE {Solving $\det\hat{\Sigma}_{WW|\mathbf{Z}}$ by Eq.(\ref{eq:009})}
\STATE {$t=1$}
\WHILE {$t\leq Ite$ or $\mid\mathbf{\Psi}^{t+1} - \mathbf{\Psi}^{t}\mid \le \epsilon$}
\STATE {$\mathbf{\Psi}^{t+1} = \mathbf{\Psi}^{t} -\beta Tr[\hat{\Sigma}^{-1}_{WW|\mathbf{Z}}\frac{\partial\hat{\Sigma}_{WW|\mathbf{Z}}}{\partial \mathbf{\Psi}}]$}
\ENDWHILE
\STATE {Computing $\mathbf{Z} = \mathbf{\Psi}^{T}\mathbf{X}$}
\STATE {Computing $Dist(z_i, z_j)$ over $\mathbf{Z}$.}
\STATE {Imputing $Y^{*}_{i}(w_i)$ via Eq.(\ref{eq:NNM}).}
\STATE {Calculating the causal effect of $W$ on $Y$.}
\RETURN{Causal effect}
\end{algorithmic}
\end{algorithm}

\section{Experiments}
\label{sec:exp}
Evaluating causal effect estimator is very challenging since we rarely have the ground truth of causal effects on real-world datasets. Following existing literature, we evaluate CESD on five datasets, including two semi-synthetic real-world datasets, including IHDP~\cite{hill2011bayesian} and Twins~\cite{louizos2017causal}; and three real-world applications, Job training~\cite{lalonde1986evaluating}, Cattaneo2~\cite{ghosh2020sufficient} and RHC~\cite{connors1996outcomes}.

The developed CESD method consists of kernel dimension reduction and NNM which are implemented by the $\mathbb{R}$ packages \emph{KDRcpp}\footnote{https://github.com/aschmu/KDRcpp} and \emph{Matching}~\cite{ho2007matching}, respectively. To evaluate the performance of CESD, we compare it against the state-of-the-art causal effect estimators including \textbf{MDM}: Mahalanobis distance matcing~\cite{rubin1979using}; \textbf{PSM}: propensity score matching with logistic regression~\cite{rubin1973matching}; \textbf{CBPS}\footnote{\url{https://cran.r-project.org/web/packages/CBPS/index.html}}: covariate balancing propensity score~\cite{imai2014covariate}; \textbf{PAW}: the set of causes of $W$ with PSM~\cite{haggstrom2018data}; \textbf{PAY}: the set of causes of $Y$ with PSM~\cite{haggstrom2018data}; \textbf{CausalForest}\footnote{\url{https://cran.r-project.org/web/packages/grf/index.html}}: Random forest regression for estimating causal effect~\cite{wager2018estimation}; \textbf{Shrinkage}\footnote{\url{https://www.stat4reg.se/software/sdrcausal}}: Inverse probability weighting estimator based on SDR for average causal effect estimation~\cite{ghosh2020sufficient} and \textbf{MSDR}: matching using SDR~\cite{luo2019matching}.

\textbf{Parameter settings}. For CausalForest, we set the number of trees to 200. For Shrinkage, MSDR, and CESD, the dimension of the reduced covariates $r$ is set to 2. For CESD, the parameters $\epsilon$, $\delta$ and $Ite$ are set to 0.0001, 5 and 20 respectively, for all datasets.

\textbf{Evaluation metrics}. We evaluate the performance of all algorithms using the root-mean-square error (RMSE) and the estimation bias (\%) (relative error) when the ground truth is available. Due to page limit, the detailed results are provided in the Supplement. In the paper, we visualize the estimated causal effects and their confidence intervals with the confidence level of 95\%.  


\subsection{Experiments on the two semi-synthetic real-world datasets}
\label{subsec:threerealapplications}
\subsubsection{IHDP}
\label{subsec:IHDP}
The IHDP dataset is an observational data from a randomized trial based on the Infant Health and Development Program (IHDP), which investigated the effects of intensive high-quality care on low-birth-weight and premature infants~\cite{hill2011bayesian}. The indicator variable, representing with/without intensive high-quality care, is used as treatment variable. IHDP consists of 747 samples with 24 pretreatment variables, among which 608 are control units (samples) and 139 are treated units. The simulated outcomes are generated by using setting ``A'' in the $\mathbb{R}$ package \emph{npci}\footnote{https://github.com/vdorie/npci}, and the ground truth of the causal effect, i.e. 4.36 is obtained by the noiseless outcome according to the same procedures suggested by Hill~\cite{hill2011bayesian}.

\begin{figure}[t]
  \centering
  \includegraphics[width=8.5cm,height=4.5cm]{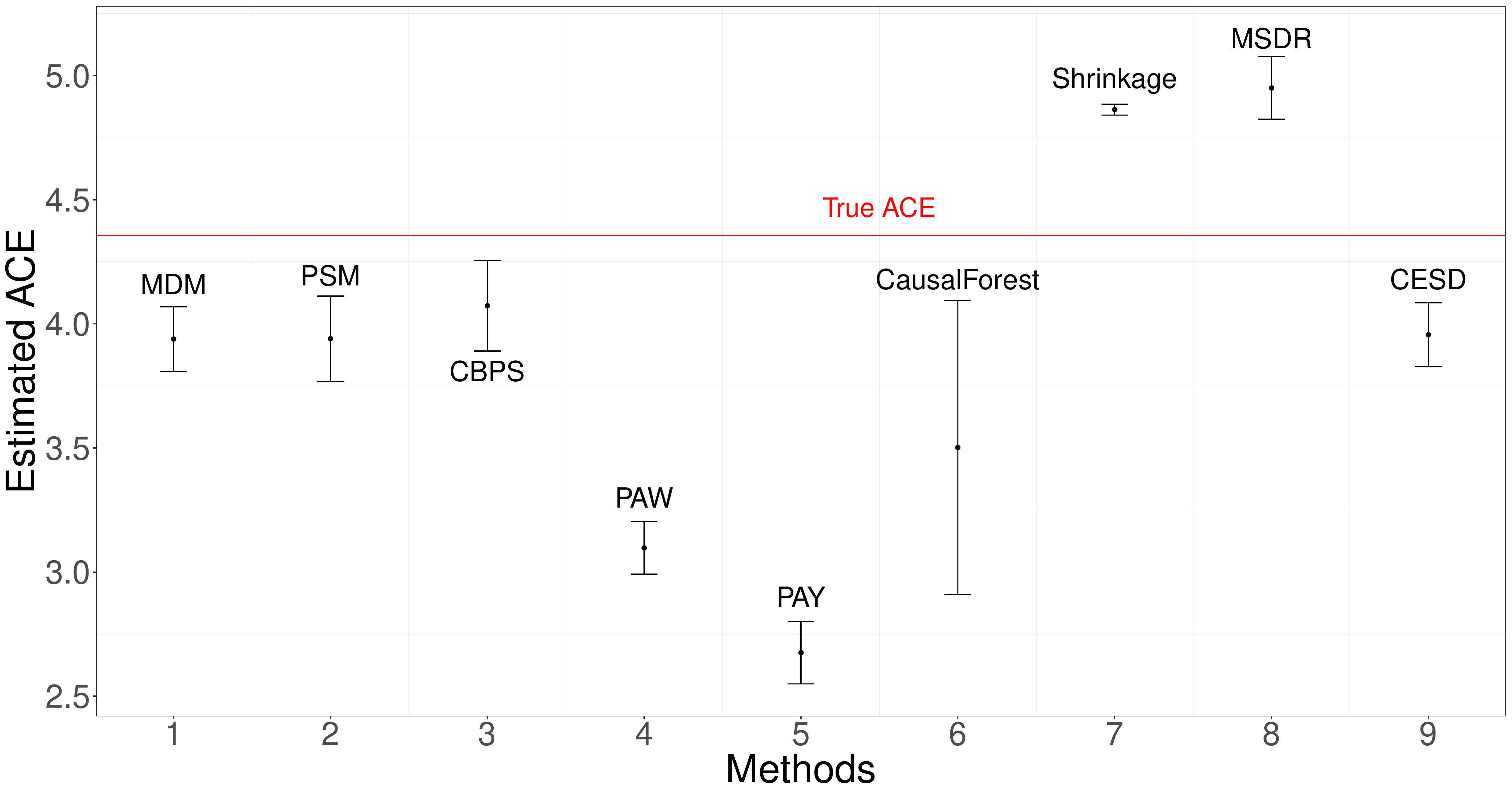}
  \caption{Estimated causal effects on the IHDP dataset w.r.t. 95\% confidence interval. The red line represents the ground truth ACE.}
  \label{fig:IHDP_plot}
\end{figure}

The experimental results of all estimators are displayed in Fig.~\ref{fig:IHDP_plot}. According to the results, these algorithms can be divided into two groups: Group I including methods whose estimates are close to the ground truth, i.e. MDM, PSM, CBPS, Shrinkage, MSDR, and CESD; the remaining methods (including PAW, PAY, and CausalForest) are in Group II. The performance of CESD is competitive with the methods in Group I and better than the methods in Group II.

\subsubsection{Twins}
\label{subsec:twins}
The Twins dataset is collected from twin births in the USA between 1989 and 1991, with infants having birth weight less than 2000g~\cite{almond2005costs}. We remove samples with missing values from the original dataset and have 4821 twin pairs left with 40 covariates. The weight of an infant is considered as the treatment variable: $W$=1 for a baby who is heavier in the twin pair; $W$=0 otherwise. The mortality after one year is the outcome. The ground truth causal effect is -0.025. To simulate an observational study, we follow Louizos et al.'s suggestion~\cite{louizos2017causal} to randomly select one of the two twins as the observed infant and hide the other by applying the setting: $W_i|x_i\sim Bern(sigmoid(\beta^{T}\mathbf{x}+\varepsilon))$, where $\mathbf{x}$ denotes the 40 other covariates, and $\beta^{T}\sim\mathcal{U}((-0.1,0.1)^{40\times1})$ and $\varepsilon\sim\mathcal{N}(0,0.1)$.

The experimental results of all methods are presented in Fig~\ref{fig:Twins_plot}. From the figure, the performance of the estimators can be divided into Group I, including PSM, CBPS, PAY, MSDR, and CESD, whose results are close to the true ACE; and Group II, including the remaining methods. We see that the performance of CESD is competitive with the methods in Group I and better than the methods in Group II.

\begin{figure}[t]
  \centering
  \includegraphics[width=8.5cm,height=4.5cm]{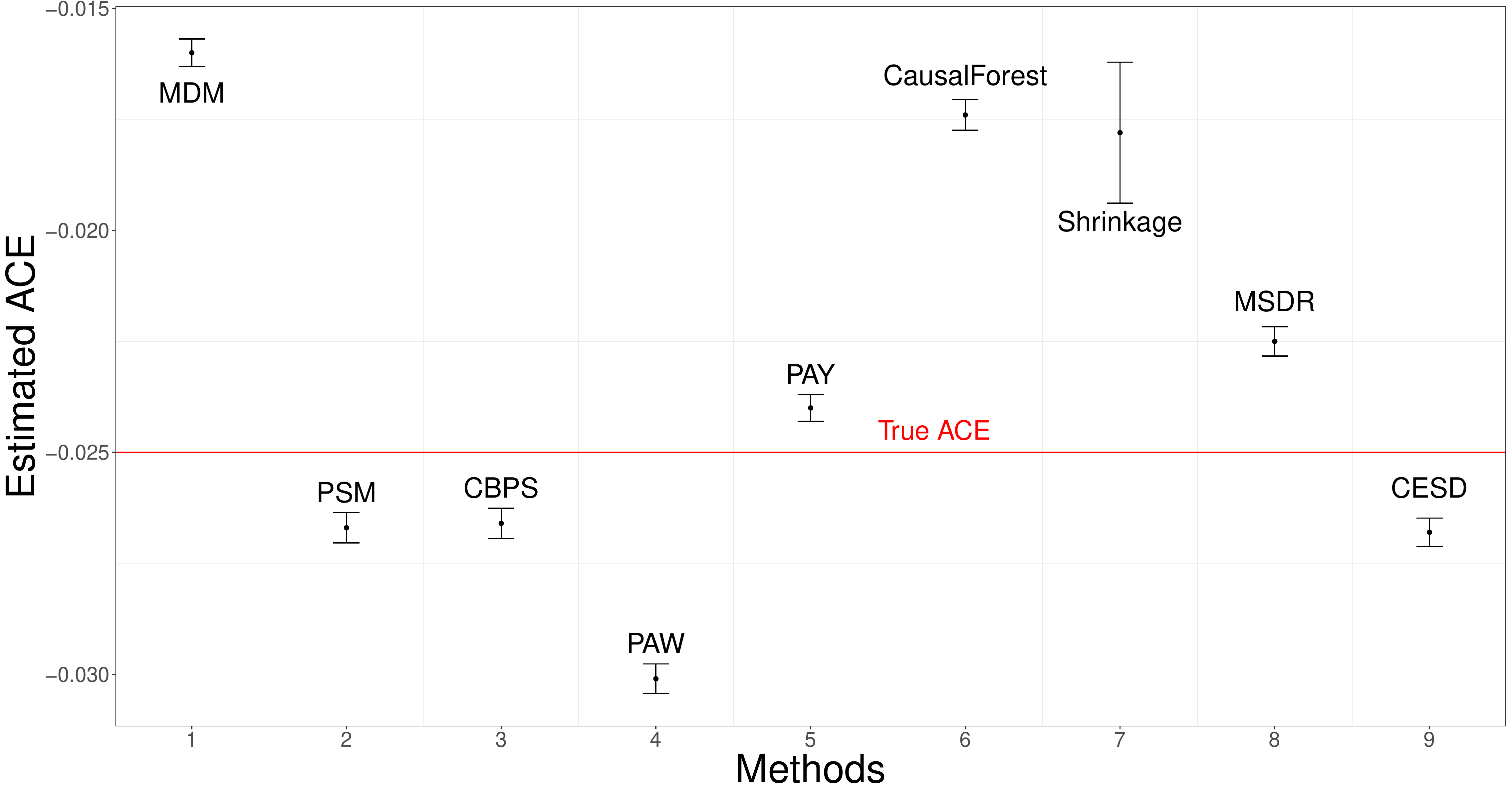}
  \caption{Estimated causal effects on the Twins dataset w.r.t. 95\% confidence interval. The red line represent the group truth ACE.}
  \label{fig:Twins_plot}
\end{figure}

\subsection{Evaluation with three real-world datasets}
\label{subsec:threerealapplications}
\subsubsection{Jobs}
\label{subsec::jobtraining}
The Job training dataset (or Jobs) is a widely used benchmark dataset in causal inference, which consists of the original LaLonde dataset (297 treated samples and 425 control samples)~\cite{lalonde1986evaluating} and the Panel Study of Income Dynamics (PSID) observational group (2490 control samples)~\cite{imai2014covariate}. There are 9 covariates, including age; schooling in years; indicators for black and Hispanic; marital status; school degree; previous earnings in 1974, 1975; and whether the 1974 earnings variable is missing. The job training status, i.e. with/without job training, is defined as the treatment variable $W$. The earning in 1978 is defined as the outcome variable $Y$. Because the dataset contains records of people taking part in the training only, as in~\cite{lalonde1986evaluating}, we estimate the $ACT$ using CESD and all comparing methods, against the ground truth $ACT$, which is \$886 with a standard error of \$448.

We draw the results of all methods in Fig.~\ref{fig:Jobs_plot}. From the figure, we see that CBPS and CESD are in Group I where the methods' estimates are fall within the empirical estimation interval, and other methods are not in the interval and hence in Group II. CESD achieves completive results with CBPS. Moreover, the estimates by Group II methods lead to the opposite conclusion, i.e. employees who participate in job training receive fewer incomes than employees who do not participate in job training.



\begin{figure}[t]
  \centering
  \includegraphics[width=8.5cm,height=4.5cm]{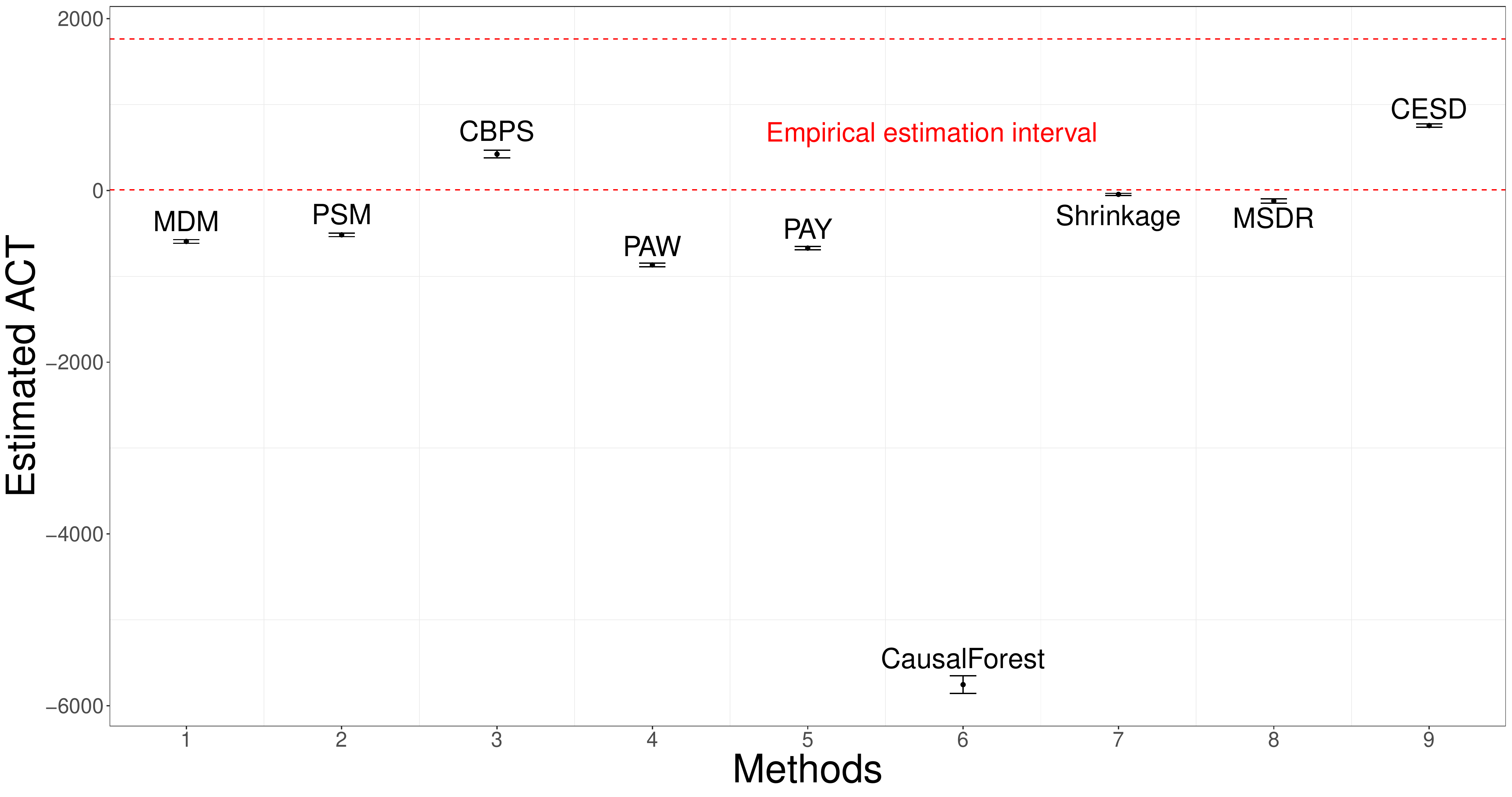}
  \caption{Estimated causal effects on the Jobs dataset. The two dotted lines denote empirical estimated interval with 95\% confident interval.}
  \label{fig:Jobs_plot}
\end{figure}

\subsubsection{Cattaneo2}
\label{subsec::catteon}
The dataset~\cite{cattaneo2010efficient}, is usually used to study the $ACE$ of maternal smoking status during pregnancy ($W$) on babies' birth weight (in grams)\footnote{\url{http://www.stata-press.com/data/r13/cattaneo2.dta}}.
Cattaneo2 consists of birth weights of 4642 singleton births in Pennsylvania, USA~\cite{almond2005costs,cattaneo2010efficient}. Cattaneo2 contains 864 smoking mothers ($W$=1) and 3778 non-smoking mothers ($W$=0). The dataset contains several covariates: mother's age, mother's marital status, an indicator for the previous infant where the newborn died, mother's race, mother's education, father's education, number of prenatal care visits, months since last birth, an indicator of firstborn infant and indicator of alcohol consumption during pregnancy. The authors~\cite{almond2005costs} found a strong negative effect of mother smoking on the weights of babies about $200g$ to $250g$ lighter for a baby with a mother smoking during pregnancy.

All results on this dataset are shown in the Fig~\ref{fig:Cattaneo2_plot}. The range of the estimated causal effects of smoking on babies' birth weight is -$285.36g$ to -$152g$. From Fig~\ref{fig:Cattaneo2_plot}, we see that only the estimated $ACE$ by CESD falls within the empirical estimated interval (-$250g$, -$200g$), i.e. the estimated effects by CESD is consistent with the original study~\cite{almond2005costs}.

\begin{figure}[t]
  \centering
  \includegraphics[width=8.5cm,height=4.5cm]{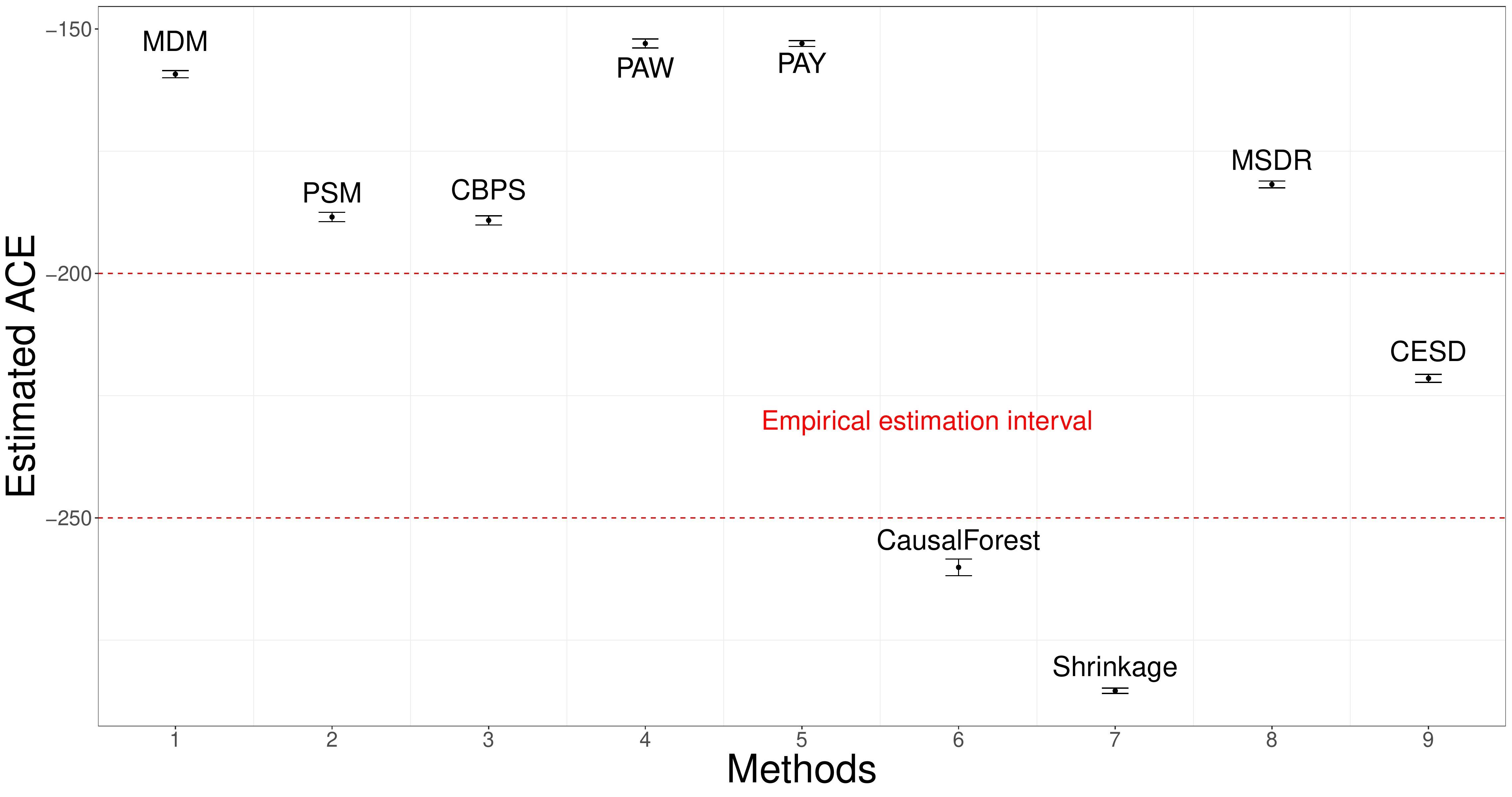}
  \caption{Estimated causal effects on the Cattaneo2 dataset w.r.t. 95\% confident interval. The two dotted lines represent empirical estimated interval (-$250g$, -$200g$).}
  \label{fig:Cattaneo2_plot}
\end{figure}

\subsubsection{Right Heart Catheterization}
\label{subsec::rhc}
Right Heart Catheterization (RHC) is a dataset from an observational study regarding a diagnostic procedure for the management of critically ill patients~\cite{connors1996effectiveness}. The RHC dataset can be downloaded from the $\mathbb{R}$ package \emph{Hmisc}\footnote{\url{https://CRAN.R-project.org/package=Hmisc}}. RHC contains the information of hospitalized adult patients from five medical centers in the USA. These hospitalized adult patients participated in the Study to Understand Prognoses and Preferences for Outcomes and Risks of Treatments (SUPPORT). The treatment $W$ indicates whether or not a patient received an RHC within 24 hours of admission. The outcome $Y$ is whether a patient died at any time up to 180 days since admission. The original RHC dataset has 5735 samples with 73 covariates. We pre-process the original data, as suggested by Loh et al.~\cite{loh2020confounder}, and the final data contains 2707 samples with 72 covariates.

The experimental results on this dataset are shown in Fig.~\ref{fig:rhc_plot}, where we can see that the result of CESD is consistent with those of PSM, CBPS, PAW, PAY, and CausalForest. The estimated causal effects by the methods indicate that applying RHC leads to higher mortality within 180 days than not applying RHC. 


\begin{figure}[t]
  \centering
  \includegraphics[width=8.5cm,height=4.5cm]{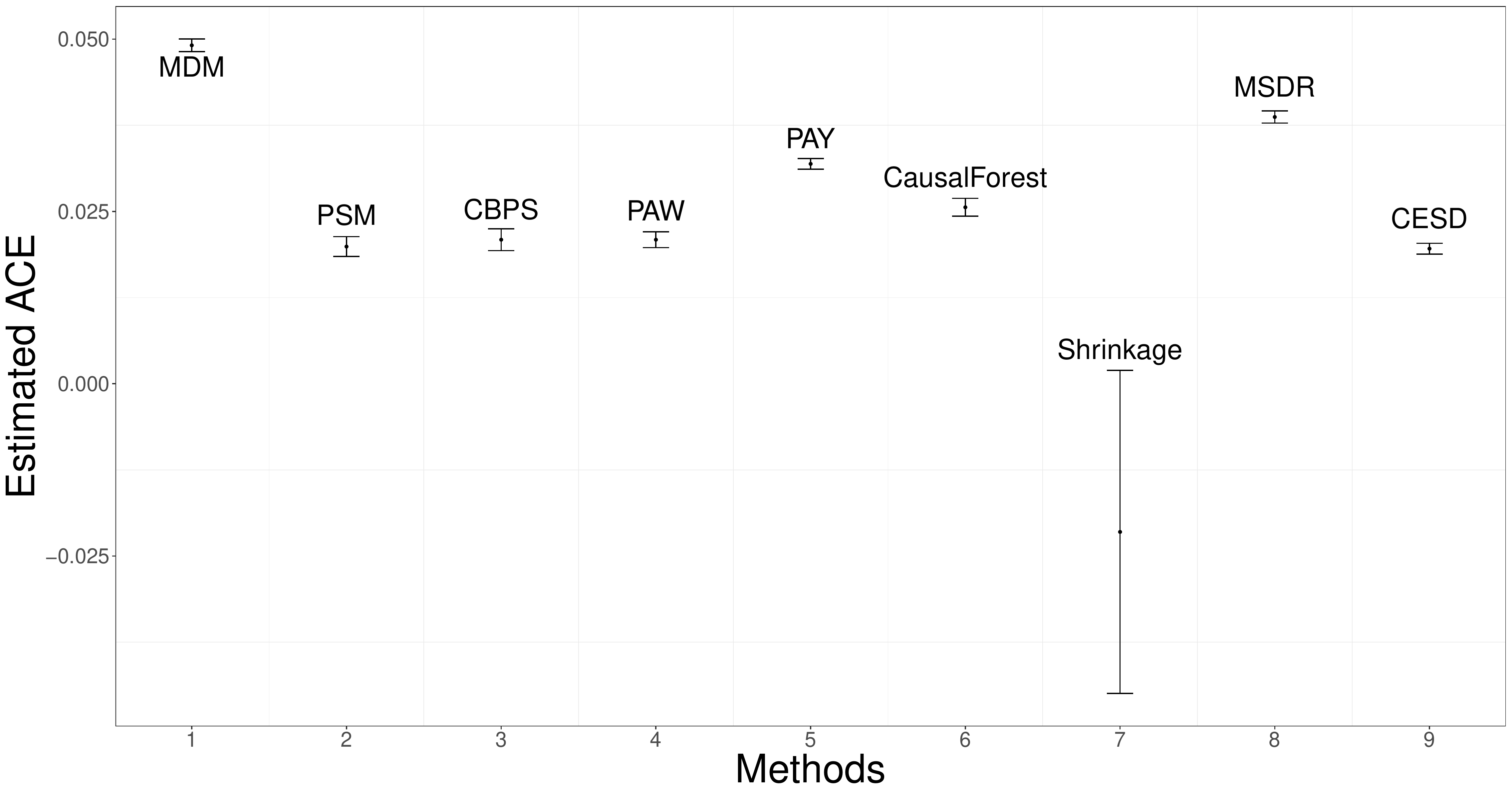}
  \caption{Estimated causal effects on the RHC dataset w.r.t. 95\% confident interval.}
  \label{fig:rhc_plot}
\end{figure}

In sum, based on all experimental results with the five datasets, we can conclude that CESD achieves the estimates which are close to true or empirically estimated causal effect values in all datasets and is consistently in the high performing method group across all the datasets, and CESD is the only method which is in the high performing group across all five datasets. The closest high performing method is CBPS from a widely used $\mathbb{R}$ package \emph{CBPS} in causal effect estimation. All these have demonstrated the robustness of the proposed method.

\subsection{The quality of matching}
To investigate further the better performance of CESD comparing with the other methods, we dig into the Cattaneo2 dataset to see the matching process, which is crucial for causal effect estimation methods which are based on propensity score and estimate propensity score using all covariates (PSM), factors of $W$ only (PAW), factors of $Y$ only (PAY), or balanced propensity score between treatment and control groups (CBPS). These methods all reduce the covariate set a one dimension propensity score for matching. We show the distributions of the estimated propensity scores in the treated and control groups in Fig.~\ref{fig:cattaneo_density} (left four sub figures). We see obvious mismatch of propensity score distributions in the two groups, and this leads to the loss of power in matching and results in a large variance in the estimated causal effect~\cite{stuart2010matching}. When the covariate set is reduced to two dimensions by kernel dimension reduction in CESD, we can see that the distributions of each reduced dimension in the treated and control groups (the curves in the right panel in Fig.~\ref{fig:cattaneo_density}) largely overlap. The overlap improves matching in causal effect estimation and this provides an explanation for the good performance of CESD.

\begin{figure}[t]
  \centering
  \includegraphics[width=8.4cm,height=5.8cm]{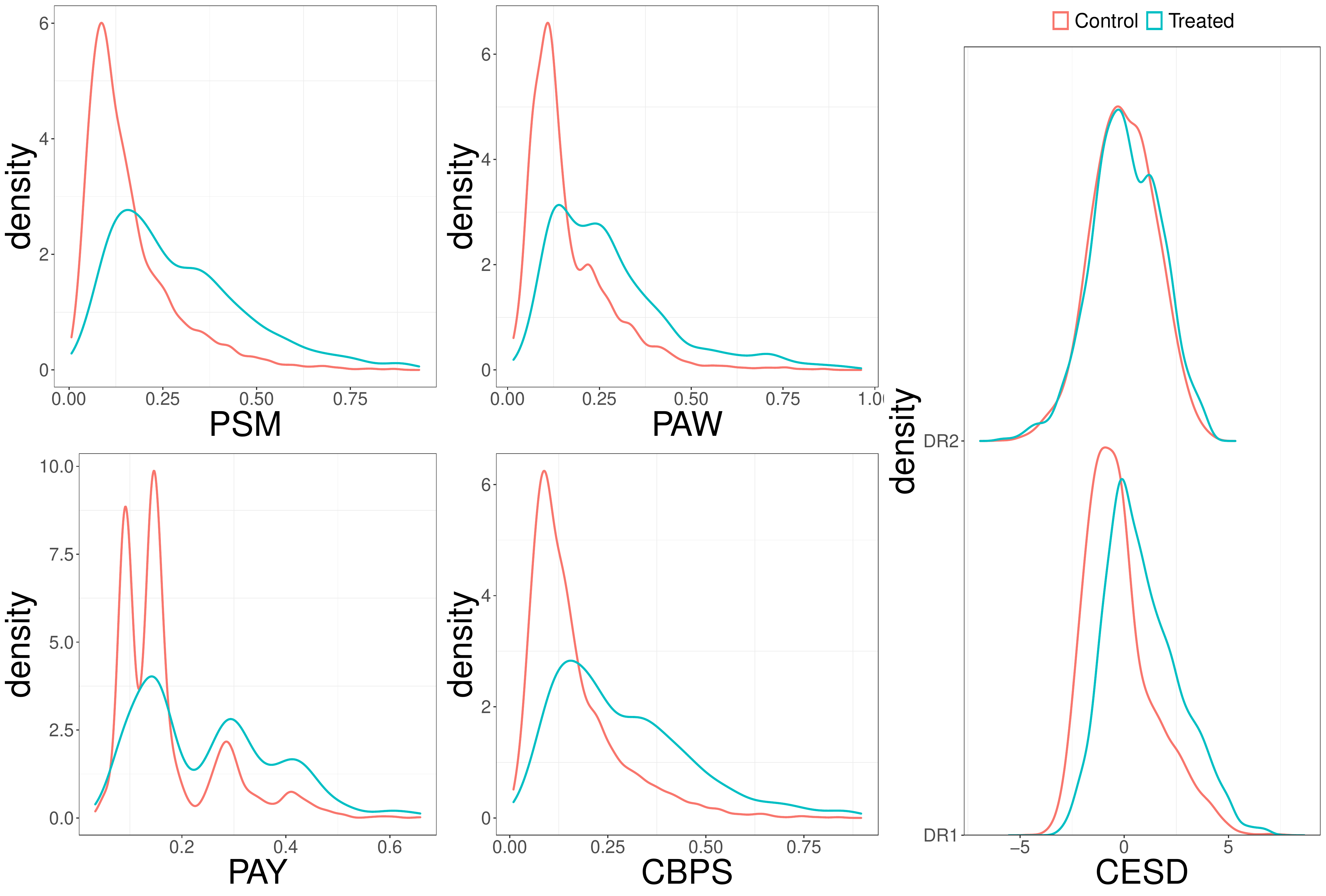}
  \caption{The distributions of propensity scores (left four) and reduced covariate dimensions, grouped by the treatment variable $W$ (red for the control group and blue for the treated group) on the Cattaneo2 dataset.}
  \label{fig:cattaneo_density}
\end{figure}

\section{Conclusion}
\label{sec:concs}
We have proposed a novel solution for average causal effect estimations through sufficient dimension reduction. In theory, we have proven the soundness of the solution where the reduced low-dimensional covariates are sufficient to remove confounding bias based on the graphical causal model, under the assumptions of pretreatment variables and unconfoundedness. We have developed a data-driven algorithm based on kernel dimension reduction, CESD, to estimate causal effects from observational data. Experimental results on two semi-synthetic real-world datasets and three real-world datasets demonstrate that CESD performs consistently very well in all five datasets in comparison with the state-of-the-art methods. This means that CESD is high performing and consistent, and is potentially useful in various areas for average causal effect estimation.

\bibliographystyle{AAAI}
\bibliography{aaai2021}  

\end{document}